\documentclass[runningheads]{llncs}
\usepackage[T1]{fontenc}
\usepackage[hyphens]{url}  
\usepackage{graphicx} 
\usepackage{soul}
\usepackage[hyphens]{url}
\usepackage[hidelinks]{hyperref}
\usepackage[utf8]{inputenc}
\usepackage{graphicx}
\usepackage{amsmath,amssymb}
\usepackage{booktabs}
\usepackage{dialogue}
\usepackage{algorithm}
\usepackage{algpseudocode}
\algnewcommand\algorithmicswitch{\textbf{switch}}
\algnewcommand\algorithmiccase{\textbf{case}}
\algnewcommand\algorithmicassert{\texttt{assert}}
\algnewcommand\Assert[1]{\State \algorithmicassert(#1)}%
\algdef{SE}[SWITCH]{Switch}{EndSwitch}[1]{\algorithmicswitch\ #1\ \algorithmicdo}{\algorithmicend\ \algorithmicswitch}%
\algdef{SE}[CASE]{Case}{EndCase}[1]{\algorithmiccase\ #1}{\algorithmicend\ \algorithmiccase}%
\algtext*{EndSwitch}%
\algtext*{EndCase}%

\newcommand{\K}{{\sf K}}
\newcommand{\B}{{\sf B}}

\renewcommand{\phi}{\varphi}

\newcommand{\counterfactual}{\ensuremath{%
  \mathrel{\scriptstyle\square\kern-1.5pt\raise0.2pt\hbox{$\mathord{\rightarrow}$}}}}

\renewenvironment{proof}{\begin{trivlist}\item\noindent{\sc Proof.}}{\hfill$\Box\hspace{2mm}$\end{trivlist}}

\newenvironment{proof-of-claim}{\noindent{\sc Proof of Claim.}}{\hfill $\Box\hspace{2mm}$\linebreak}

\usepackage[utf8]{inputenc}

\begin{document}

\title{Dynamic Logic of Trust-Based Beliefs}

\author{Junli Jiang\inst{1} \and
Pavel Naumov\inst{2} \and
Wenxuan Zhang\inst{3}}
\authorrunning{Jiang, Naumov, and Zhang}
%
\institute{Institute of Logic and Intelligence, Southwest University, China \\
\email{walk08@swu.edu.cn}
\and University of Southampton, United Kingdom\\
\email{p.naumov@soton.ac.uk}
\and Independent Scholar, United States\\
\email{wenxzhang.work@gmail.com}\\
}

\maketitle

\begin{abstract}

Traditionally, an agent's beliefs would come from what the agent can see, hear, or sense. In the modern world, beliefs are often based on the data available to the agents. In this work, we investigate a dynamic logic of such beliefs that incorporates public announcements of data. The main technical contribution is a sound and complete axiomatisation of the interplay between data-informed beliefs and data announcement modalities. We also describe a non-trivial polynomial model checking algorithm for this logical system.
\end{abstract}





\section{Introduction}


\begin{dialogue}
\speak{The Associated Press} { Breaking: Two Explosions in the White House and Barack Obama is injured.} 
\end{dialogue}

This was a tweet by the leading news agency at 1:07 pm on April 23, 2013. An hour after the tweet, the hackers from the Syrian Electronic Army claimed responsibility for the tweet~\cite{f13wp}. In this paper, we represent this Associate Press (AP) tweet as a data variable $t$ of type {string}, whose values could differ between epistemic worlds. 
Let us consider the moment before the AP tweet became public. Anyone with the knowledge of the tweet's content at that moment, like the Syrian hackers, would know that the tweet is about two explosions. We denote this by:
$$
\K_t(\mbox{``the AP tweet is about two explosions''}).
$$
In general, we consider arbitrary {\em sets} of data variables. We call them {\em datasets}. We allow data variables to be of an arbitrary type. Informally, $\K_X\phi$ means that anyone who knows dataset $X$ would know $\phi$. More formally, $\K_X\phi$ means that statement $\phi$ is true in all worlds in which all data variables in dataset $X$ have the same values as in the current world. We call $\K_X$ the {\em data-informed knowledge} modality. This modality for sets of Boolean variables is introduced in~\cite{gls15jair}. For arbitrary data variables, it is proposed in~\cite{bv21jpl}. The term ``data-informed knowledge'' is coined in~\cite{jn22ai}. 

Van Eijck, Gattinger, and Wang~\cite{vgw17icla} propose modality $[X]\phi$ that stands for ``statement $\phi$ is true after set $X$ of data variables is publicly announced''. In our example,
$$
[t]\K_\varnothing(\mbox{``the AP tweet is about two explosions''})
$$
because once the value of $t$ is publicly announced, no additional information is needed to know that the tweet is about two explosions. 

Our Twitter story does not end here. Just {\em two minutes} after the tweet, S\&P 500 declined 0.95\%, wiping \$136.5 billion in stock value~\cite{f13wp}. Joseph Greco, formerly of Meridian Equity Partners, placed the blame on computerised trading algorithms that monitor news sites and trigger trades based on predetermined rules~\cite{l18mw}. In other words, a {\em public announcement} on the {\em trusted} AP Twitter account made the automated trading agents to form a {\em belief} that the US economy is about to decline. Based on this belief, the agents started the sell-off.

\begin{figure}[ht]
\centering
\vspace{0mm}
\scalebox{0.35}{\includegraphics{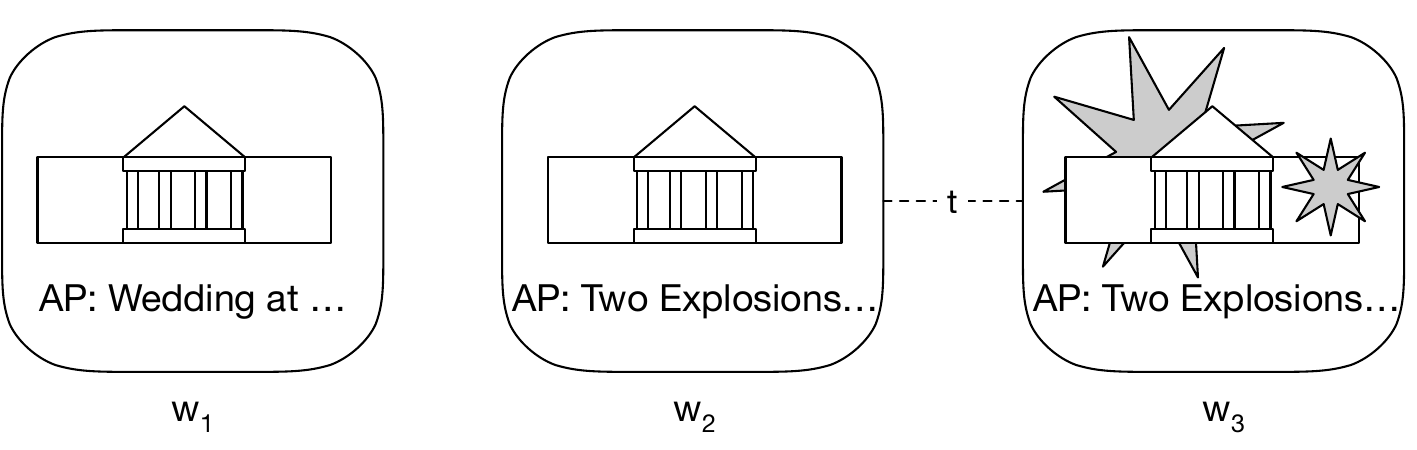}}
\caption{\footnotesize Trustworthiness model. Data variable $t$ is trustworthy in worlds $w_1$ and $w_3$ and is not trustworthy in world $w_2$.}\label{white house figure}
\vspace{-2mm}
\end{figure}
We previously proposed a formal logical system for reasoning about such {\em trust-based beliefs}~\cite{jn24synthese-trust}.
The semantics of that logic is defined using trustworthiness models.
Such a model for our example is depicted in Figure~\ref{white house figure}. This model has three possible worlds: $w_1$, $w_2$, and $w_3$. In the first world, the White House is safe, and the Associated Press is tweeting about a wedding. In the second world, the White House is also safe, but the AP tweet incorrectly reports two explosions. In world $w_3$, the tweet about two explosions is correct. The actual world is $w_2$.

A trustworthiness model specifies for each possible world which data variables are {\em trustworthy} in that world. Informally, a data variable is trustworthy in a world if it reflects the state of affairs in this world. In our example, data variable $t$ is trustworthy in worlds $w_1$ and $w_3$ but not in world $w_2$. After the tweet becomes public, an agent can distinguish worlds $w_1$ and $w_2$, but it still cannot distinguish the current world $w_2$ from possible world $w_3$. 

Note that a {\em trustworthy} variable is not necessarily {\em trusted}, and a trusted variable is not necessarily trustworthy.
We say that {\em the trust in a data variable informs a belief in a statement if the statement holds in all worlds indistinguishable from the current world in which the data variable is trustworthy}. In our example, in world $w_2$, the trust in $t$ informs the belief that the US economy is about to decline:
\begin{equation}\label{29-may-a}
    [t]\B^t_\varnothing(\mbox{``the US economy is about to decline''}).
\end{equation}
This is because the only world indistinguishable from world $w_2$  is world $w_3$ in which data variable $t$ is trustworthy and the US economy is about to decline.

In the above example, the belief is based only on publicly announced data. In general, following~\cite{jn24synthese-trust}, we write $\B^T_X\phi$ to denote that if dataset $T$ is trusted, then dataset $X$, together with the publicly announced data, informs the belief in statement $\phi$. 

Observe that statement~(\ref{29-may-a}) is true in worlds $w_2$ and $w_3$ but not in world $w_1$. Let us now consider the situation {\em before} the tweet $t$ becomes public. At that moment, an agent (with access only to public data) cannot distinguish any of the worlds $w_1$, $w_2$, and $w_3$. If the agent trusts the AP tweets, then the agent believes in the statements true in all worlds indistinguishable from $w_2$ in which data variable $t$ is trustworthy. That is, the trust in data variable $t$ leads to the belief in all statements that are true in world $w_1$ and $w_3$. As noted above, statement~(\ref{29-may-a}) is not true in world $w_1$. Thus, in the current world~$w_2$,

$$
\neg \B^t_\varnothing[t]\B^t_\varnothing(\mbox{``the US economy is about to decline''}).
$$
In other words, before the tweet goes public, the agent trusting the AP (but not knowing yet the content of the tweet $t$) does not believe that after the tweet goes public, it will expect the economy to decline. At the same time, anyone who knows the content of tweet $t$ before it becomes public (like, for example, the Syrian hackers), would consider only $w_2$ and $w_3$ as possible worlds. Note that statement~(\ref{29-may-a}) is true in both of these worlds. Thus,

$$
\K_t[t]\B^t_\varnothing(\mbox{``the US economy is about to decline''}).
$$
Informally, anyone who knows $t$ would also know that, after the tweet goes public, people who trust the AP will believe that the US economy is about to decline.

\section{Related Literature}\label{Related Literature section}

Public {\em data} announcement modality $[X]\phi$ has its roots in the Public Announcement Logic (PAL)~\cite[Chapter 4]{dhk07}. PAL extends the language of the epistemic logic with a public {\em statement} announcement modality $[\phi]\psi$ that means ``if truthful statement $\phi$ is publicly announced, then statement $\psi$ will become true''. The AP tweet from the introduction {\em cannot} be treated as a public announcement in the PAL setting because PAL only allows announcements of true statements. 

Multiple extensions of Public Announcement Logic are suggested. 
W{\'a}ng and {\AA}gotnes
 add the distributed knowledge modality to it~\cite{wa13synthese}.
{\AA}gotnes, Balbiani, van Ditmarsch, and Seban,
propose a group announcement modality $\langle G\rangle\phi$ that means ``group $G$ can announce certain facts, individually known to the members of the group, after which statement $\phi$ will be true''~\cite{abvs10jal}. Although modality $\langle G\rangle\phi$ states that $\phi$ will become true after an announcement by group $G$, it does not require $\phi$ to remain true after further announcements are made by agents outside of group $G$. However, such requirement is imposed by modality $\langle[G]\rangle\phi$ introduced in~\cite{ga17tark}.

In \cite{jn24synthese-trust}, we review the existing literature on logics of beliefs and trust and discuss the connection between modality $\B^T_X$ and the previously studied operators capturing trust and beliefs. Besides~\cite{jn24synthese-trust}, the most relevant previous work is~\cite{pgca19lori}, which proposes a logical system that describes the interplay between trust, beliefs, and public group announcements. In that paper, trust is semantically modelled through set $T^w_a$ of all agents whom agent $a$ trusts in state $w$ and beliefs are defined using belief bases. As public announcements are made, agent $a$ updates the set $T^w_a$ of agents that it trusts based on its belief base. 
Thus, in~\cite{pgca19lori}, beliefs define trust, while in \cite{jn24synthese-trust} and the current paper, trust defines beliefs. The syntax of the system in~\cite{pgca19lori} includes trust atomic proposition $T_{a,b}$ (agent $a$ trusts agent $b$), belief modality $B_a\phi$ (agent $a$ beliefs in statement $\phi$), and group public announcement. They consider only announcements of beliefs expressed in a restricted language, which is {\em not} closed with respect to Boolean connectives. Unlike us, they do not consider data-informed beliefs and announcements of values. They suggest an axiomatisation of their system but do not provide a proof of its completeness.

\section{Contribution and Outline}


In this paper, we study the interplay between modalities $\K_X$, $\B^T_X$, and $[X]$. The connection between modalities $\K_X$ and $\B^T_X$ is straightforward: $\K_X\phi\equiv\B^\varnothing_X\phi$. The connection between modalities $\B^T_X$ and $[X]$ is much less trivial, and this connection is the main focus of the current work. To describe this connection, we propose a non-trivial Commutativity axiom:  $[Y]\B^T_X\phi \leftrightarrow \B^T_{Y\cup X}[Y]\phi$ and show that all properties of the interplay between these two modalities follow from this axiom and studied earlier individual properties of modalities $\B^T_X$ and $[X]$.

To give the formal semantics for the language containing modalities $\B^T_X$ and $[X]$, we propose to define satisfaction as a ternary relation $w,U\Vdash\phi$ between a possible world, a dataset, and a formula. Neither of the papers~\cite{vgw17icla,jn22ijcai-trust} defines semantics through such a relation. This approach also forced us to state the truth lemma in a new form, using formula $[U]\phi$ rather than formula $\phi$, as done in both of the works mentioned above. Similarly, the statement and the proof of Lemma~\ref{child exists} are modified to use $[U]\B^T_X\phi\notin F(w)$ instead of $\B^T_X\phi\notin F(w)$. 

Finally, we give a non-trivial model checking algorithm for our language. Works~\cite{vgw17icla,jn22ijcai-trust} do not discuss model checking. We are also not aware of any follow-up works discussing model checking for the logical system in~\cite{vgw17icla}.

The rest of the paper is structured as follows. In the next section, we define the class of models that we later use to give the semantics of our logical system. In Section~\ref{s and s section}, we give its syntax and formal semantics. In Section~\ref{Belief Revision section}, we discuss whether a public announcement of data can change a belief to the opposite. We list the axioms and the inference rules of our system and prove their soundness in Section~\ref{Axioms section}. 
Section~\ref{Completeness section} proves the completeness of our system. In Section~\ref{Model Checking section}, we propose a polynomial time model checking algorithm for our logical system. Section~\ref{Conclusion section} concludes.

\section{Trustworthiness Model}\label{Trustworthiness Model section}

In this section, we formally define trustworthiness models. We fix the set of data variables $V$ and the set of atomic propositions.

In the introduction, we assumed that data variables have values. For example, tweet $t$ had string values such as ``Two explosions ...''. However, as we will see in Definition~\ref{sat}, the specific values of variables are not important for the semantics of the modality $\B^T_X$. It is only important whether the values of all variables in dataset $X$ in two different epistemic worlds are the same or not.
Thus, to keep our presentation succinct, we use equivalence relation $\sim_x$ for each data variable $x\in V$ as the primitive component of a trustworthiness model in the formal definition below. Informally, $w_1\sim_x w_2$ if variable $x$ has the same value in worlds $w_1$ and $w_2$.

\begin{definition}\label{model}
A tuple $(W,\{\sim_x\}_{x\in V},\{\mathcal{T}_w\}_{w\in W},\pi)$ is called a {trustworthiness model} if 
\begin{enumerate}
    \item $W$ is a (possibly empty) set of worlds,
    \item $\sim_x$ is an ``indistinguishability'' equivalence relation on set $W$ for each data variable $x\in V$,
    \item $\mathcal{T}_w\subseteq V$ is a set of data variables that are ``trustworthy'' in world $w\in W$,  
    \item $\pi(p)\subseteq W\times \mathcal{P}(V)$ for each atomic proposition $p$.
\end{enumerate}
\end{definition}
In our introductory example, $V=\{t\}$, set $W$ is the set $\{w_1,w_2,w_3\}$, relation $\sim_t$ is the reflexive and symmetric closure of the relation $\{(w_2,w_3)\}$, $\mathcal{T}_{w_1}=\mathcal{T}_{w_3}=\{t\}$, and $\mathcal{T}_{w_2}=\varnothing$.

In the standard account of Public Announcement Logic (PAL), it is assumed that the validity of each atomic proposition only depends on the world and does not depend on the announcement made in this world so far~\cite[ch. 4]{dhk07}. For example, atomic proposition $p$ in that system can not represent the statement ``formula $\phi$ has already been publicly announced''. As a result, PAL contains the Atomic Permanence axiom $[\phi]p\leftrightarrow (\phi \to p)$ that captures the fact that the validity of an atomic proposition cannot be changed by an announcement. Because the principle $[\phi]\psi\leftrightarrow (\phi\to \psi)$ is not valid for an arbitrary formula $\psi$, the substitution inference rule is not admissible in PAL.

This paper takes a more general approach under which an atomic proposition, just like any other formula in our language, represents a statement about the current world and all public announcements made in this world so far. For example, an atomic proposition can represent the statement ``the AP tweeted about two explosions that did not happen''. 
Item 4 of Definition~\ref{model} captures this by defining set $\pi(p)$ to be a set of all pairs $(w,U)$ such that atomic proposition $p$ is true in state $w\in W$ when the set of all publicly announced variables is exactly $U$.

\section{Syntax and Semantics}\label{s and s section}

The language $\Phi$ of our system is defined by the grammar:
$$
\phi::=p\;|\;\neg\phi\;|\;\phi\to\phi\;|\;\B^T_X\phi\;|\;[X]\phi,
$$
where $p$ is an atomic proposition and $X,T\subseteq V$ are datasets.
We read formula $\B^T_X\phi$ as ``if dataset $T$ is trusted, then dataset $X$ informs the belief $\phi$'' and formula $[X]\phi$ as ``statement $\phi$ is true after a public announcement of dataset $X$''.  Connective $\leftrightarrow$ and constant $\bot$ are defined through $\neg$ and $\to$ in the usual way.

Next, we define the semantics of our logical system. In most of the works on modal logic, including~\cite{jn24synthese-trust}, the semantics of a modal logic is defined through a binary satisfaction relation $w\Vdash\phi$ between a world $w\in W$ and a formula $\phi\in \Phi$. Informally, it means that formula $\phi$ is true in world $w$. To account for the public announcements of data variables, in this paper, we define satisfaction as a ternary relation $w,U\Vdash\phi$ between a world $w\in W$, a dataset $U\subseteq V$, and a formula $\phi\in\Phi$.
Informally, $w,U\Vdash\phi$ means that statement $\phi$ is true in world $w$ when the set of publicly announced variables is exactly $U$. Throughout the paper, we write $w\sim_X u$ if $w\sim_x u$ for each data variable $x\in X$.

\begin{definition}\label{sat}
For any world $w\in W$ of any trustworthiness model $(W,\{\sim_x\}_{x\in V},\{\mathcal{T}_w\}_{w\in W},\pi)$, any dataset $U\subseteq V$, and any formula $\phi\in \Phi$, the satisfaction relation $w,U\Vdash\phi$ is defined as follows:
\begin{enumerate}
    \item $w,U\Vdash p$ if $(w,U)\in \pi(p)$,
    \item $w,U\Vdash\neg\phi$ if $w,U\nVdash\phi$,
    \item $w,U\Vdash \phi\to\psi$ if $w,U\nVdash \phi$ or $w,U\Vdash \psi$,
    \item $w,U\Vdash \B^T_X\phi$ if $w',U\Vdash \phi$ for each world $w'\in W$ such that $w\sim_{X\cup U} w'$ and $T\subseteq \mathcal{T}_{w'}$,
    \item $w,U\Vdash [X]\phi$ if $w,U\cup X\Vdash \phi$.
\end{enumerate}
\end{definition}
Note that, in item~4 of the above definition, we use relation $\sim_{X\cup U}$ because dataset $U$ is publicly announced and, thus, the values of variables in this dataset are available while inferring belief $\phi$ from the dataset $X$. 

In the special case when set $T$ is empty, the modality $\B^\varnothing_X\phi$ means that formula $\phi$ is true in all worlds indistinguishable from the current world by the dataset $X$ and the set  $U$ of all publicly announced variables. In the introduction, we denoted this modality by $\K_X\phi$. It is easy to see that it satisfies all standard S5 properties.

Finally, note that in this paper we consider public announcements of {\em values} of data variables. One might also consider public announcements of {\em trustworthiness} of data variables. We leave this type of public announcements for future research.

\section{Belief Revision}\label{Belief Revision section}

Under the semantics proposed in the previous section, it is possible that $w,U\Vdash \B^T_X\bot$. Indeed, this is true if there is no world $u\in W$ such that $w\sim_X u$ and dataset $T$ is trustworthy in world $u$. In such a situation, the set of trusted variables $T$ needs to be changed just as it is done with beliefs in the traditional belief revision literature~\cite[ch. 3]{dhk07}. The exact procedure of how this should be done in real-world applications is important, but it is outside of the scope of our work. By introducing a trust parameter into the belief modality, we are able to separate the question of what should be trusted from what should be believed based on the existing trust. 

In PAL, only true statements can be announced, but it is a well-known observation that once a true statement is announced it might become false~\cite[ch. 4]{dhk07}. In our system, variables do not change values. Thus, their public announcements do not affect their values. One naturally can ask if a public announcement of variables can affect data-informed beliefs:

\noindent
{\bf Question:} {\em is it possible for formulae $\B^T_X\phi$ and $[X]\neg\B^T_X\phi$ 
to be true at the same time? What about formulae $\B^T_X\phi$ and $[X]\B^T_X\neg\phi$?}

The answer to both parts of this question is \textit{yes}. Indeed, assume that the language contains just a single atomic proposition $p$ and a single data variable $x$, and consider a trustworthiness model with a single world $w$. The choice of set $\mathcal{T}_{w} $ is not important. Suppose that $\pi(p)=\{(w,\varnothing)\}$. In other words, atomic proposition $p$ represents the statement ``nothing has been announced''. First, observe that
$w,\varnothing\Vdash p$ by item~1 of Definition~\ref{sat}. Thus, $w,\varnothing\Vdash \B^\varnothing_{x} p$ by item~4 of Definition~\ref{sat} because $w$ is the only world in this model. Next, observe that $w,\{x\}\nVdash p$ by item~1 of Definition~\ref{sat}. Thus, $w,\{x\}\nVdash \B^\varnothing_{x} p$ by item~4 of Definition~\ref{sat}. Hence, $w,\varnothing\Vdash [x]\neg\B^\varnothing_{x} p$ by items~2 and 5 of Definition~\ref{sat}. This answers the first part of the question. To answer the second, one can similarly show that $w,\varnothing\Vdash [x]\B^\varnothing_{x} \neg p$.

A PAL-traditionalist might find the above answer unsatisfactory because it exploits the fact that the truth value of atomic proposition $p$ depends not only on the world, but also on the announcements made in this world. As we discussed in Section~\ref{Trustworthiness Model section}, such atomic propositions are not allowed in the standard version of PAL. For such readers, we have another, a bit more complicated, example that uses a ``permanent'' atomic proposition. In this example, we assume the language contains a single atomic proposition $p$ and two data variables: $x$ and $y$. The trustworthiness model for this example is depicted in Figure~\ref{belief revision figure}. 
\begin{figure}[ht]
\begin{center}
\vspace{-2mm}
\scalebox{0.5}{\includegraphics{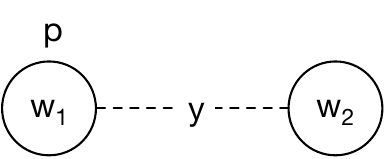}}
\caption{\footnotesize Trustworthiness model.}\label{belief revision figure}
\vspace{-4mm}
\end{center}
\vspace{0mm}
\end{figure}

It has two worlds, $w_1$ and $w_2$, distinguishable by variable $x$ but indistinguishable by variable $y$. The choice of sets $\mathcal{T}_{w_1} $ and $\mathcal{T}_{w_2}$ is not important.
Atomic proposition $p$ is (permanently) true in world $w_1$ and is (permanently) false in world $w_2$. It is easy to see that
$w_1,\varnothing\Vdash \B^\varnothing_x\neg\B^\varnothing_y p$;
$w_1,\varnothing\Vdash [x]\neg\B^\varnothing_x\neg\B^\varnothing_yp$;
and
$w_1,\varnothing\Vdash [x]\B^\varnothing_x\neg\neg\B^\varnothing_y p$.
This answers the question (take $\phi$ to be the statement $\neg\B^\varnothing_y p$).

\section{Axioms}\label{Axioms section}

In addition to propositional tautologies in language $\Phi$, our {\em Dynamic Logic of Trust-Based Beliefs} contains the axioms listed below. 

\begin{enumerate}
     \item Truth: $\B^\varnothing_X\phi\to\phi$,
    \item Distributivity:\\
    $\B^T_X(\phi\to\psi)\to(\B^T_X\phi\to\B^T_X\psi)$,\\
    $[X](\phi\to\psi)\to([X]\phi\to[X]\psi)$,
    \item Negative Introspection of Beliefs: $\neg\B^T_X\phi\to\B^\varnothing_X\neg\B^T_X\phi$,
     \item Monotonicity: $B^T_X\phi\to\B^{T'}_{X'}\phi$,
       where $T\subseteq T'$, $X\subseteq X'$,
    \item Trust: $\B_X^T(\B_Y^T\phi\to\phi)$,
    \item Combination: $[X][Y]\phi\leftrightarrow [X\cup Y]\phi$,
    
    \item Commutativity: $[Y]\B^T_X\phi \leftrightarrow \B^T_{Y\cup X}[Y]\phi$,
    
    \item Duality: $\neg[X]\phi\leftrightarrow[X]\neg\phi$,
 
    \item Empty Announcement: $[\varnothing]\phi\leftrightarrow \phi$.
\end{enumerate}

We write $\vdash\phi$ and say that formula $\phi$ is a {\em theorem} if $\phi$ is provable from the above axioms using the Modus Ponens and the Necessitation inference rules:
$$
\dfrac{\phi, \phi\to\psi}{\psi}
\hspace{15mm}
\dfrac{\phi}{\B^T_X\phi}
\hspace{15mm}
\dfrac{\phi}{[X]\phi}.
$$
In addition to the unary relation $\vdash\phi$, we also consider a binary relation $F\vdash\phi$. We write $F\vdash\phi$ if formula $\phi$ is derivable from the {\em theorems} of our logical system and the set of additional assumptions $F$ using the Modus Ponens inference rule only. Note that statement $\varnothing\vdash\phi$ is equivalent to $\vdash\phi$. We say that a set of formulae $F$ is {\em inconsistent} if $F\vdash\phi$ and $F\vdash\neg\phi$ for some formula $\phi\in\Phi$.

\begin{theorem}[strong soundness]\label{strong soundness}
For any dataset $U\subseteq V$, any world $w$ of a trustworthiness model, any set of formulae $F\subseteq \Phi$, and any formula $\phi\in\Phi$, if $w,U\Vdash f$ for each formula $f\in F$ and $F\vdash \phi$, then $w,U\Vdash \phi$.
\end{theorem}
The soundness of the
Truth, the Distributivity, the Monotonicity, the Combination, the Duality, and the Empty Announcement axioms is straightforward. 
Below, we prove the soundness of the Negative Introspection of Beliefs, the Trust, and the Commutativity axioms.
\begin{lemma}
If $w, U \nVdash \B^T_X\phi$, then $w, U \Vdash \B^\varnothing_X\neg\B^T_X\phi$.
\end{lemma}
\begin{proof}
The assumption $w, U \nVdash \B^T_X\phi$ by item~4 of Definition~\ref{sat} implies that  there is a world $w' \in W$ such that 

\begin{equation}\label{3-Dec-c}
   w\sim_{U \cup X} w',
\end{equation}
\begin{equation}\label{3-Dec-a}
    T\subseteq\mathcal{T}_{w'},
\end{equation}
and
\begin{equation}\label{3-Dec-b}
    w' , U \nVdash\phi.
\end{equation}

Consider any world $v\in W$ such that $w\sim_{U \cup X} v$. By item~4 of Definition~\ref{sat}, it suffices to show that $v, U \nVdash \B^T_X\phi$. Assume the opposite. Then, $v, U \Vdash \B^T_X\phi$. Note that, because $\sim_{U \cup X}$ is an equivalence relation, statement~(\ref{3-Dec-c}) and the assumption $w\sim_{U \cup X} v$ imply that  $v\sim_{U \cup X} w'$ . Therefore, $w', U  \Vdash\phi$ by item~4 of Definition~\ref{sat} and statement~(\ref{3-Dec-a}), which contradicts statement~(\ref{3-Dec-b}).
\end{proof}

\begin{lemma}
$w, U \Vdash \B^T_X(\B^T_Y\phi\to \phi)$.
\end{lemma}
\begin{proof}
Consider any world $w'\in W$ such that $w\sim_{U\cup X} w'$ and $T\subseteq\mathcal{T}_{w'}$. By item~4 of Definition~\ref{sat}, it suffices to show that $w', U \Vdash \B^T_Y\phi\to \phi$. Suppose that $w', U \Vdash \B^T_Y\phi$. By item~3 of Definition~\ref{sat}, it is enough to prove that $w', U \Vdash\phi$. 

Note that $w'\sim_{U \cup Y} w'$ because relation $\sim_{U \cup Y}$ is reflexive. Also, $T\subseteq\mathcal{T}_{w'}$ by the choice of world $w'$. Then, the assumption $w', U \Vdash \B^T_Y\phi$ implies that $w', U \Vdash\phi$ by item~4 of Definition~\ref{sat}.
\end{proof}

\begin{lemma}
$w,U\Vdash [Y]\B^T_X\phi$ iff $w,U\Vdash \B^T_{X\cup Y}[Y]\phi$.
\end{lemma}
\begin{proof}
$(\Rightarrow):$ Assume that $w,U\nVdash \B^T_{X\cup Y}[Y]\phi$.
Thus, by item~4 of Definition~\ref{sat}, there exists a world $v\in W$ such that $w\sim_{U\cup X\cup Y}v$ , $T\subseteq\mathcal{T}_{v}$, and $v,U\nVdash [Y]\phi$.
Hence, by item~5 of Definition~\ref{sat}, it follows that $v,U\cup Y\nVdash \phi$.
Then, item~4 of Definition~\ref{sat} implies that $w,U\cup Y\nVdash \B^T_X\phi$.
Therefore, $w,U\nVdash [Y]\B^T_X\phi$ again by item~5 of Definition~\ref{sat}.

\vspace{1mm}
\noindent $(\Leftarrow):$ Suppose $w,U\nVdash [Y]\B^T_X\phi$. 
Thus, it follows that $w,U\cup Y\nVdash \B^T_X\phi$ by item~5 of Definition~\ref{sat}.
Hence, by item~4 of Definition~\ref{sat}, there exists a world $v\in W$ such that $w\sim_{U\cup Y\cup X}v$ , $T\subseteq\mathcal{T}_{v}$, and $v,U\cup Y\nVdash \phi$. 
Then, $v,U\nVdash [Y]\phi$ by item~5 of Definition~\ref{sat}.
Thus, by item~4 of Definition~\ref{sat} implies that $w,U\nVdash \B^T_{X\cup Y}[Y]\phi$.
\end{proof}

We conclude this section with two lemmas used in the proof of completeness.

\begin{lemma}\label{positive introspection lemma}
$\vdash \B^T_X\phi\to\B^\varnothing_X\B^T_X\phi$. 
\end{lemma}
\begin{proof}
Formula $\B^\varnothing_X\neg\B^T_X\phi\to\neg\B^T_X\phi$ is an instance of the Truth axiom. Thus, $\vdash \B^T_X\phi\to\neg\B^\varnothing_X\neg\B^T_X\phi$ by contraposition. Hence, taking into account the following instance $\neg\B^\varnothing_X\neg\B^T_X\phi\to\B^\varnothing_X\neg\B^\varnothing_X\neg\B^T_X\phi$ of  the Negative Introspection axiom,
we have 
\begin{equation}\label{pos intro eq 2}
\vdash \B^T_X\phi\to\B^\varnothing_X\neg\B^\varnothing_X\neg\B^T_X\phi.
\end{equation}

At the same time, formula $\neg\B^T_X\phi\to\B^\varnothing_X\neg\B^T_X\phi$ is also an instance of the Negative Introspection axiom. Thus, by contraposition, $\vdash \neg\B^\varnothing_X\neg\B^T_X\phi\to \B^T_X\phi$. Hence, 
 by the Necessitation inference rule, $\vdash \B^\varnothing_X(\neg\B^\varnothing_X\neg\B^T_X\phi\to \B^T_X\phi)$. Thus, the Distributivity axiom and the Modus Ponens inference rule imply that
$
  \vdash \B^\varnothing_X\neg\B^\varnothing_X\neg\B^T_X\phi\to \B^\varnothing_X\B^T_X\phi.
$
 The latter, together with statement~(\ref{pos intro eq 2}), implies the statement of the lemma by propositional reasoning.
\end{proof}

\begin{lemma}\label{super distributivity}
If $\phi_1,\dots,\phi_n\vdash\psi$, then $\B^T_X\phi_1,\dots,\B^T_X\phi_n\vdash\B^T_X\psi$.
\end{lemma}
\begin{proof}
By deduction lemma applied $n$ times, the assumption $\phi_1,\dots,\phi_n\vdash\psi$ implies that
$$
\vdash\phi_1\to(\phi_2\to\dots(\phi_n\to\psi)\dots).
$$
Thus, by the Necessitation inference rule,
$$
\vdash\B^T_X(\phi_1\to(\phi_2\to\dots(\phi_n\to\psi)\dots)).
$$
Hence, by the Distributivity axiom and the Modus Ponens rule,
$$
\vdash\B^T_X\phi_1\to\B^T_X(\phi_2\to\dots(\phi_n\to\psi)\dots).
$$
Then, again by the Modus Ponens rule,
$$
\B^T_X\phi_1\vdash\B^T_X(\phi_2\to\dots(\phi_n\to\psi)\dots).
$$
Therefore, $\B^T_X\phi_1,\dots,\B^T_X\phi_n\vdash\B^T_X\psi$ by applying the previous steps $(n-1)$ more times.
\end{proof}

\section{Completeness}\label{Completeness section}

The proof of completeness is divided into three parts. First, we use tree construction to define the canonical model. Then, we prove the key properties of this model, including the truth lemma. Finally, we use the truth lemma to prove completeness.

\subsection{Canonical Model}

Following~\cite{jn24synthese-trust}, we use the tree construction to define the canonical trustworthiness model $M(T_0,F_0)=(W,\{\sim_x\}_{x\in V},\{\mathcal{T}_w\}_{w\in W},\pi)$ for any dataset $T_0\subseteq V$ and any maximal consistent set of formula $F_0\subseteq \Phi$. The origins of the tree construction can be traced back to the proof of completeness for the distributed knowledge modality~\cite{fhv92jacm}. 

\begin{definition}\label{canonical W}
Set $W$ of worlds  is the set of all sequences
$T_0,F_0,X_1,T_1,F_1,\dots$, $X_n,T_n,F_n$ such that $n\ge 0$ and, for each $i$ where $0\le i\le n$, 
\begin{enumerate}
    \item $X_i,T_i\subseteq V$ are datasets,
  \item $F_i$ is a maximal consistent set of formulae such that 
  \begin{enumerate}
      \item $\psi\in F_i$ for each formula $\B^\varnothing_{X_i}\psi\in F_{i-1}$, if $i>0$,
      \item $\B^{T_i}_Y\phi\to \phi\in F_i$ for each dataset $Y\subseteq V$ and each formula $\phi\in\Phi$.
  \end{enumerate}  
\end{enumerate}
\end{definition}
If $w',w\in W$ are any two worlds such that 
\begin{eqnarray*}
w'&=&T_0,F_0,\dots,X_{n-1},T_{n-1},F_{n-1},\\
w&=&T_0,F_0,\dots,X_{n-1},T_{n-1},F_{n-1},X_n,T_n,F_n,
\end{eqnarray*}
then we say that worlds $w'$ and $w$ are {\em adjacent}. The adjacency relation forms a tree structure on set $W$. We say that the {\em edge} connecting nodes $w'$ with node $w$ is {\em labelled} with all variables in dataset $X_n$. We also say that {node} $w$ is {\em labelled} with the pair $T_n,F_n$. Finally, by $T(w)$ and $F(w)$ we denote sets $T_n$ and $F_n$, respectively.

\begin{figure}[ht]
\begin{center}
\vspace{-2mm}
\scalebox{0.5}{\includegraphics{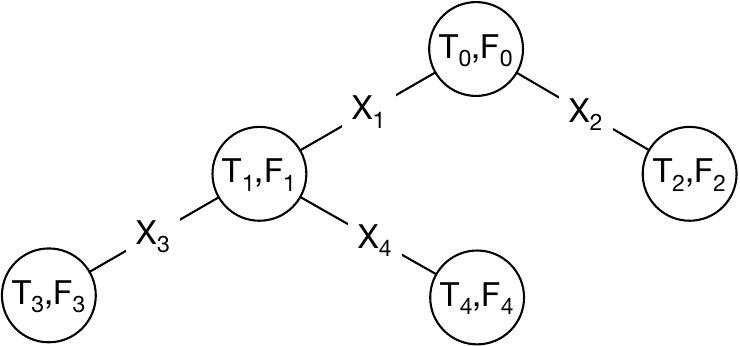}}
\vspace{-1mm}
\caption{\footnotesize Fragment of tree $W$.}\label{canonical tree figure}
\vspace{-5mm}
\end{center}
\vspace{0mm}
\end{figure}
Figure~\ref{canonical tree figure} visualises the tree structure on the set $W$. In this figure, world $w'=T_0,F_0,X_1,T_1,F_1$ is adjacent to the world $w=T_0,F_0,X_1,T_1,F_1,X_4,T_4,F_4$. The edge between nodes $w'$ and $w$ is labelled by all variables in set $X_4$.

\begin{definition}\label{canonical sim}
For any worlds $u,v\in W$ and any data variable $x\in V$, let $u\sim_x v$ if every edge along the unique simple path between node $u$ and node $v$ is labelled with variable~$x$.
\end{definition}
\begin{lemma}
Relation $\sim_x$ is an equivalence relation on set $W$ for each data variable $x\in V$.
\end{lemma}

\begin{definition}\label{canonical T}
$\mathcal{T}_w=T(w)$.
\end{definition}

\begin{definition}\label{canonical pi}
For any atomic proposition $p$,
$$\pi(p)=\{(w,U)\in W\times \mathcal{P}(V)\;|\;[U]p\in F(w)\}.$$
\end{definition}


\subsection{Properties of the Canonical Model}


\begin{lemma}\label{single transfer lemma}
For any formula $\B^T_Y\phi\in\Phi$ and any worlds
\begin{eqnarray*}
w'&=&T_0,F_0,\dots,X_{n-1},T_{n-1},F_{n-1},\\
w&=&T_0,F_0,\dots,X_{n-1},T_{n-1},F_{n-1},X_n,T_n,F_n,
\end{eqnarray*}
if $Y\subseteq X_n$, then $\B^T_Y\phi\in F(w')$ \,iff\, $\B^T_Y\phi\in F(w)$.
\end{lemma}
\begin{proof} $(\Rightarrow):$ Suppose
$\B^T_Y\phi\in F(w')$. Thus,
$\B^T_Y\phi\in F_{n-1}$. Then, by Lemma~\ref{positive introspection lemma} and the Modus Ponens inference rule, 
$F_{n-1}\vdash \B^\varnothing_Y\B^T_Y\phi$. 
Hence, 
$F_{n-1}\vdash \B^\varnothing_{X_n}\B^T_Y\phi$
by the assumption $Y\subseteq X_n$ of the lemma, the Monotonicity axiom, and the Modus Ponens inference rule.
Then, $\B^\varnothing_{X_n}\B^T_Y\phi\in F_{n-1}$ because $F_{n-1}$ is a maximal consistent set. Thus, $\B^T_Y\phi\in F_{n}$ by item~2(a) of Definition~\ref{canonical W}. Therefore, $\B^T_Y\phi\in F(w)$. 

\vspace{1mm}\noindent$(\Leftarrow):$ 
Suppose
$\B^T_Y\phi\notin F(w')$.
Then, $\B^T_Y\phi\notin F_{n-1}$. Thus, $\neg\B^T_Y\phi\in F_{n-1}$ because $F_{n-1}$ is a maximal consistent set of formulae. Hence, $F_{n-1}\vdash \B^\varnothing_Y\neg\B^T_Y\phi$ by the Negative Introspection axiom and the Modus Ponens inference rule. 
Thus, 
$F_{n-1}\vdash \B^\varnothing_{X_n}\neg\B^T_Y\phi$ by the assumption $Y\subseteq X_n$ of the lemma, the Monotonicity axiom, and the Modus Ponens inference rule.
Then, again because set $F_{n-1}$ is maximal, $\B^\varnothing_{X_n}\neg\B^T_Y\phi\in F_{n-1}$. Thus, $\neg\B^T_{X_n}\phi\in F_{n}$ by item~2(a) of Definition~\ref{canonical W}.
Hence,
$\B^T_Y\phi\notin F_{n}$, because set $F_n$ is consistent. Therefore, $\B^T_Y\phi\notin F(w)$.
\end{proof}

\begin{lemma}\label{child all}
For any worlds $w,u\in W$ and any formula $\B^T_X\phi\in F(w)$, if $w\sim_X u$ and $T\subseteq \mathcal{T}_u$, then $\phi\in F(u)$. 
\end{lemma}
\begin{proof}
By Definition~\ref{canonical sim}, the assumption $w\sim_X u$ implies that each edge along the unique simple (without self-intersections) path between nodes $w$ and $u$ is labelled with each variable in dataset $X$. Then, the assumption of the lemma
$\B^T_X\phi\in F(w)$ implies 
$
    \B^T_X\phi\in F(u)
$
by applying Lemma~\ref{single transfer lemma} to each edge along this path. Note that the assumption $T\subseteq \mathcal{T}_u$ of the lemma implies that $T\subseteq T(u)$ by Definition~\ref{canonical T}. Thus, 
$
    F(u) \vdash \B^{T(u)}_X\phi 
$
by the Monotonicity axiom and the Modus Ponens inference rule. 
Hence,
$
F(u) \vdash \phi
$
by item~2(b) of Definition~\ref{canonical W} and the Modus Ponens inference rule. Therefore, $\phi\in F(u)$ because the set $F(u)$ is maximal.
\end{proof}

The next lemma significantly differs from the proof in~\cite{jn24synthese-trust} and other similar proofs in modal logic because it prefixes formulae with $[U]$. Its proof relies on the use of the Commutative axiom, which captures the interplay between beliefs and public announcement modalities.

\begin{lemma}\label{child exists}
For any world $w\in W$ and any formula $[U]\B^T_X\phi\notin F(w)$, there exists a world $w'\in W$ such that $w\sim_{U\cup X} w'$, $T\subseteq\mathcal{T}_{w'}$, and $[U]\phi\notin F(w')$.
\end{lemma}
\begin{proof}
Consider the following set of formulae
\begin{eqnarray}\label{G definition}
G&=&\{\neg[U]\phi\}\cup \{\psi\;|\;\B^\varnothing_{U\cup X}\psi\in F(w)\}\nonumber\\
&&\cup \{\B^T_Y\chi\to \chi\;|\;Y\subseteq V, \chi\in\Phi\}
\end{eqnarray}
\begin{claim}
Set $G$ is consistent.
\end{claim}
\begin{proof-of-claim}
Suppose the opposite. Then, there are formulae $\chi_1,\dots,\chi_n\in\Phi$, datasets $Y_1,\dots,Y_n\subseteq V$, and formulae
\begin{equation}\label{19nov-a}
\B^\varnothing_{U\cup X}\psi_1,\dots,\B^\varnothing_{U\cup X}\psi_m\in F(w)
\end{equation}
such that
$$
\B^T_{Y_1}\chi_1\to \chi_1,\dots,\B^T_{Y_n}\chi_n\to \chi_n,
\psi_1,\dots,\psi_m\vdash [U]\phi. 
$$
Thus, by Lemma~\ref{super distributivity},
\begin{eqnarray*}
&&\hspace{-5mm}\B^T_{U\cup X}(\B^T_{Y_1}\chi_1\to \chi_1),\dots,\B^T_{U\cup X}(\B^T_{Y_n}\chi_n\to \chi_n),\\
&&\hspace{+20mm}\B^T_{U\cup X}\psi_1,\dots,\B^T_{U\cup X}\psi_m\vdash \B^T_{U\cup X}[U]\phi. 
\end{eqnarray*}
Hence, 
$
\B^T_{U\cup X}\psi_1,\dots,\B^T_{U\cup X}\psi_m\vdash \B^T_{U\cup X}[U]\phi
$
by the Trust axiom applied several times. 
Thus, by the Monotonicity axiom and the Modus Ponens inference rule also applied several times,
$
\B^\varnothing_{U\cup X}\psi_1,\dots,\B^\varnothing_{U\cup X}\psi_m\vdash \B^T_{U\cup X}[U]\phi
$.
Hence,
$$
F(w)\vdash \B^T_{U\cup X}[U]\phi
$$
due to statement~(\ref{19nov-a}).
Thus,
$
F(w)\vdash [U]\B^T_X\phi
$
by the Commutativity axiom and propositional reasoning.
Then, $[U]\B^T_X\phi\in F(w)$ because the set $F(w)$ is maximal, which contradicts the assumption of the lemma.
\end{proof-of-claim}

Define $G'$ be any maximal consistent extension of set $G$. 
Assume that $w=T_0,F_0,\dots,X_n,T_n,F_n$.
Consider sequence
\begin{equation}\label{w' choice}
    w'=T_0,F_0,\dots,X_n,T_n,F_n,U\cup X,T,G'.
\end{equation}
Observe that $w'\in W$ by Definition~\ref{canonical W}, equation~(\ref{G definition}), and the choice of set $G'$ as an extension of set $G$. Also, note that $w\sim_{U \cup X} w'$ by Definition~\ref{canonical sim} and equation~(\ref{w' choice}). Finally, $T=T(w')=\mathcal{T}_{w'}$ by equation~(\ref{w' choice}) and Definition~\ref{canonical T}. 

To finish the proof of the lemma, note that $\neg[U]\phi\in G\subseteq G'=F(w')$ by equation~(\ref{G definition}), the choice of $G'$ as an extension of $G$, and equation~(\ref{w' choice}). Then, $[U]\phi\notin F(w')$ because the set $F(w')$ is consistent. 
\end{proof}

The proof of the next ``truth'' lemma is substantially more complicated than the proof of the truth lemma in~\cite{jn24synthese-trust} because of the presence of $[U]$ prefix. 

\begin{lemma}\label{truth lemma}
$w,U\Vdash\phi$ iff $[U]\phi\in F(w)$,
for each world $w\in W$, each dataset $U\subseteq V$, and each formula $\phi\in\Phi$.
\end{lemma}
\begin{proof}
We prove the statement by induction on the complexity of formula $\phi$. 

\vspace{1mm}
Suppose that formula $\phi$ is an atomic proposition $p$. Note that $w,U\Vdash p$ iff $(w,U)\in \pi(p)$ by item~1 of Definition~\ref{sat}. At the same time, $(w,U)\in \pi(p)$ iff $[U]p\in F(w)$ by Definition~\ref{canonical pi}. Therefore, $w,U\Vdash p$ iff $[U]p\in F(w)$.

\vspace{1mm}
Suppose that formula $\phi$ has the form $\neg\psi$.

\vspace{1mm}
\noindent$(\Rightarrow):$ 
Assume $w,U\Vdash \neg\psi$. 
Then, $w,U\nVdash \psi$ by item~2 of Definition~\ref{sat}.
Hence, $[U]\psi\notin F(w)$ by the induction hypothesis.
Thus, $\neg[U]\psi\in F(w)$ because set $F(w)$ is maximal.
Then, $F(w)\vdash [U]\neg\psi$ by the Duality axiom and the Modus Ponens inference rule.
Therefore, $[U]\neg\psi\in F(w)$  because set $F(w)$ is maximal.

\vspace{1mm}
\noindent$(\Leftarrow):$ 
Assume $[U]\neg\psi\in F(w)$. 
Thus, $F(w)\vdash \neg [U]\psi$ by the Duality axiom and propositional reasoning.
Then, because set $F(w)$ is consistent, $[U]\psi\notin F(w)$.
Hence, $(w,U)\nVdash \psi$ by the induction hypothesis.
Therefore, $(w,U)\Vdash \neg\psi$ by item~2 of Definition~\ref{sat}.

\vspace{1mm}
Suppose that formula $\phi$ has the form $\psi_1\to\psi_2$.

\vspace{1mm}
\noindent$(\Rightarrow):$
Assume  $w,U\Vdash\psi_1\to\psi_2$.
Thus, either $w,U\nVdash\psi_1$  or $w,U\Vdash\psi_2$ by item~3 of Definition~\ref{sat}. We consider these two cases separately.

\vspace{1mm}
\noindent {\em Case I}: $w,U\nVdash\psi_1$.
Then, $[U]\psi_1\notin F(w)$ by the induction hypothesis.
Hence, $\neg[U]\psi_1\in F(w)$ because set $F(w)$ is maximal.
Thus, by the Duality axiom and propositional reasoning,
\begin{equation}\label{jan25-d}
    F(w)\vdash [U]\neg\psi_1.
\end{equation}
Note that the formula $\neg\psi_1\to(\psi_1\to\psi_2)$ is a propositional tautology.
Then, $\vdash[U](\neg\psi_1\to(\psi_1\to\psi_2))$ by the Necessitation inference rule.
Hence, $\vdash[U]\neg\psi_1\to[U](\psi_1\to\psi_2)$ by the Distributivity axiom and the Modus Ponens inference rule.
Thus, $F(w)\vdash [U](\psi_1\to\psi_2)$ by statement~(\ref{jan25-d}) and the Modus Ponens inference rule.
Then, $[U](\psi_1\to\psi_2)\in F(w)$ because set $F(w)$ is maximal.

\vspace{1mm}
\noindent {\em Case II}: $w,U\Vdash\psi_2$.
Then, by the induction hypothesis,
\begin{equation}\label{jan25-e}
    [U]\psi_2\in F(w).
\end{equation}
Note that the formula $\psi_2\to(\psi_1\to\psi_2)$ is a propositional tautology. Then, $\vdash [U](\psi_2\to(\psi_1\to\psi_2))$ by the Necessitation inference rule. Hence, $\vdash [U]\psi_2\to[U](\psi_1\to\psi_2)$ by the Distributivity axiom and the Modus Ponens inference rule.
Thus, $F(w)\vdash [U](\psi_1\to\psi_2)$ by statement~(\ref{jan25-e}) and the Modus Ponens inference rule.
Therefore, because set $F(w)$ is maximal, $[U](\psi_1\to\psi_2)\in F(w)$.

\vspace{1mm}
\noindent$(\Leftarrow):$ 
Assume $[U](\psi_1\to\psi_2)\in F(w)$.
Then, by the Distributivity axiom and the Modus Ponens inference rule, $F(w)\vdash [U]\psi_1\to[U]\psi_2$.
Hence, $[U]\psi_1\to[U]\psi_2\in F(w)$ because set $F(w)$ is maximal.
Thus, by the Modus Ponens inference rule,
if $[U]\psi_1\in F(w)$, then $F(w)\vdash [U]\psi_2$.
Hence, because set $F(w)$ is maximal, if $[U]\psi_1\in F(w)$, then $[U]\psi_2\in F(w)$.
Thus, by the induction hypothesis, if $w,U\Vdash\psi_1$, then $w,U\Vdash\psi_2$.
Therefore, $w,U\Vdash\psi_1\to\psi_2$ by item~3 of Definition~\ref{sat}.

\vspace{1mm}
Suppose that formula $\phi$ has the form $\B^T_X\psi$. 

\vspace{1mm}
\noindent$(\Rightarrow):$ Assume that $[U]\B^T_X\psi\notin F(w)$. Thus, by Lemma~\ref{child exists}, there is a world $w'\in W$ such that $w\sim_{U\cup X} w'$, $T\subseteq\mathcal{T}_{w'}$, and $[U]\psi\notin F(w')$. Then, $w',U\nVdash\psi$ by the induction hypothesis. Therefore, $w,U\nVdash\B^T_X\psi$ by item~4 of Definition~\ref{sat} and the assumptions $w\sim_{U\cup X} w'$ and $T\subseteq\mathcal{T}_{w'}$. 

\vspace{1mm}
\noindent$(\Leftarrow):$ Assume that $[U]\B^T_X\psi\in F(w)$. Consider any world $w'$ such that $w\sim_{U\cup X} w'$ and $T\subseteq\mathcal{T}_{w'}$. By item~4 of Definition~\ref{sat}, it suffices to show that $w',U\Vdash \psi$. 
Indeed, the assumption $[U]\B^T_X\psi\in F(w)$ implies $F(w)\vdash \B^T_{U\cup X}[U]\psi$ by the Commutativity axiom and propositional reasoning.
Thus, because set $F(w)$ is maximal, $\B^T_{U\cup X}[U]\psi\in F(w)$. 
Hence, $[U]\psi\in F(w')$ by Lemma~\ref{child all} and the assumptions $w\sim_{U\cup X} w'$ and $T\subseteq\mathcal{T}_{w'}$.
Thus, $w',U\Vdash \psi$ by the induction hypothesis.

\vspace{1mm}
Finally, suppose that formula $\phi$ has the form $[X]\psi$. 

\vspace{1mm}
\noindent
By item~5 of Definition~\ref{sat}, the statement $w,U\Vdash [X]\psi$ is equivalent to the statement $w,U\cup X\Vdash \psi$. By the induction hypothesis,  $w,U\cup X\Vdash \psi$ iff $[U\cup X]\psi\in F(w)$. Note that, because the set $F(w)$ is maximal, the statement $[U\cup X]\psi\in F(w)$ is equivalent to $F(w)\vdash [U\cup X]\psi$. By the Combination axiom and propositional reasoning, $F(w)\vdash [U\cup X]\psi$ iff $F(w)\vdash [U][X]\psi$. The statement $F(w)\vdash [U][X]\psi$ is equivalent to $[U][X]\psi\in F(w)$ also because the set $F(w)$ is maximal.
\end{proof}

\subsection{Completeness: Final Step}

\begin{theorem}[strong completeness]
For any set of formulae $F\subseteq \Phi$ and any formula $\phi\in\Phi$, if $F\nvdash \phi$, then there is a world $w$ of a trustworthiness model and a dataset $U\subseteq V$ such that $w,U\Vdash f$ for each formula $f\in F$ and $w,U\nVdash \phi$.
\end{theorem}
\vspace{-3mm}
\begin{proof}
It follows from the assumption $F\nvdash \phi$ that the set $F\cup\{\neg\phi\}$ is consistent. Consider any maximal consistent extension $F_0$ of this set and the canonical model $M(\varnothing,F_0)$. 

Let us observe that the sequence $(\varnothing,F_0)$ is a world of this canonical model. By Definition~\ref{canonical W}, it is enough to prove that $\B^\varnothing_Y\psi\to\psi\in F_0$ for each dataset $Y\subseteq V$ and each formula $\psi\in\Phi$. Observe that the statement is true by the Truth axiom and the maximality of set $F_0$.

Next, note that $\neg\phi\in F_0$. Hence, $F_0\vdash [\varnothing]\neg\phi$ by the Empty Announcement axiom and propositional reasoning. Hence, $F_0\vdash \neg[\varnothing]\phi$ by the Duality axiom and propositional reasoning. Thus, $[\varnothing]\phi\notin F_0$ because set $F_0$ is consistent. Then, $(\varnothing,F_0),\varnothing\nVdash \phi$ by Lemma~\ref{truth lemma}.

Finally, consider any formula $f\in F$. To finish the proof of the theorem, it suffices to show that $(\varnothing,F_0),\varnothing\Vdash f$. Indeed, $F_0\vdash f$ because $f\in F\subseteq F_0$. Thus, $F_0\vdash [\varnothing]f$ by the Empty Announcement axiom and propositional reasoning. Therefore, $(\varnothing,F_0),\varnothing\Vdash f$ by Lemma~\ref{truth lemma}.
\end{proof}

\section{Model Checking}\label{Model Checking section}

In this section, we present and analyse a model checking algorithm for our logical system. The model checking algorithm is a Boolean function $sat(w_0,U_0,\phi_0)$ that, for any given  world $w_0$, dataset $U_0$, and formula $\phi_0$, returns the value true iff  $w_0,U_0\Vdash\phi_0$. We assume that the set $V$ of all data variables and the set $W$ of worlds are finite. 

Note that if the set $V$ is fixed, then a polynomial model checking algorithm can be easily constructed using the dynamic programming technique. Recall that this technique recursively computes the return value of a {\em function} using an {\em array}. In our case, the dynamic programming could be used to fill in a three-dimensional Boolean array $sat[]$ in such a way that the value of $sat[w,U,\phi]$ is equal to the value of $sat(w,U,\phi)$  for each world $w\in W$, each dataset $U\subseteq V$, and each subformula $\phi$ of formula $\phi_0$. Since the set $V$ is fixed, it has a fixed number of subsets $U\subseteq V$. Thus, the array $sat[]$ has a polynomial size. As a result, the execution time of the algorithm is also polynomial.

The situation is significantly different if the set $V$ is not fixed. In this case, the number of datasets $U\subseteq V$ is exponential. Thus, a straightforward application of the above approach would result in the array $sat[]$ having an exponential size, so filling such an array would also take exponential time. In the rest of this section, we modify the dynamic programming algorithm to work in polynomial time even when set $V$ is given as a part of the input.

To understand the idea behind our version of the algorithm, it will be convenient to think about computing $sat(w,U,\phi)$ as computing the value of formula $\phi$ in world $w$ under {\em environment} $U$. Informally, the recursive description of the function $sat(w,U,\phi)$ is already given in Definition~\ref{sat}. 
Note that when function $sat(w,U,\phi)$ makes a recursive call, the environment can change. For example, computing the value of $sat(w,U,[X][Y]p)$ requires the computation of the value of $sat(w,U\cup X,[Y]p)$.
But even before that, one needs to compute $sat(w,U\cup X\cup Y,p)$.

\begin{figure}[ht]
\vspace{-2mm}
\hrule\vspace{1mm}
\begin{algorithmic}
\footnotesize\sf
\Procedure{H}{$U, \phi$}
            \Switch{$\phi$}
            \Case{$\phi$ is an atomic proposition}
                \State \Return $[(U, \phi)]$
            \EndCase
            \Case{$\phi$ has the form $\neg\psi$} 
                \State \Return $\Call{H}{U, \psi}+[(U, \phi)]$
            \EndCase
            \Case{$\phi$ has the form $\psi\to\chi$} 
                \State \Return $\Call{H}{U, \psi}+\Call{H}{U, \chi}+[(U, \phi)]$
            \EndCase
            \Case{$\phi$ has the form $\B^T_X\psi$} 
                \State \Return $\Call{H}{U, \psi}+[(U, \phi)]$
            \EndCase
            \Case{$\phi$ has the form $[X]\psi$} 
                \State \Return $\Call{H}{U\cup X, \psi}+[(U,  \phi)]$
            \EndCase
            \EndSwitch
\EndProcedure
\end{algorithmic}
\vspace{1mm}\hrule
\caption{Helper function.}\label{helper function H}
\vspace{-7mm}
\end{figure}

The key to our efficient implementation of function $sat(w,U,\phi)$ is recursive function ${\sf H}(U,\phi)$ that computes all subformulae of formula $\phi$ and the environments in which these subformulae are evaluated by the recursive calls of function $sat(w,U,\phi)$. For example, function ${\sf H}(U,[X][Y]p)$ will return the list of pairs: $[(U\cup X\cup Y,p)$, $(U\cup X,[Y]p)$, $(U,[X][Y]p)]$. The pseudo-code for function ${\sf H}(U,\phi)$ is given in Figure~\ref{helper function H}. In this pseudo-code, the symbol $+$ denotes the append operation on lists. 
Our code guarantees that if the list produced by function ${\sf H}(U,\phi)$ contains a pair $(U',\psi)$, then all subformulae of $\psi$ with their  environments will appear on that list {\em before} pair $(U',\psi)$. This is important later for the calculation of $sat(w,U,\phi)$. 
Furthermore, note that the size of the list returned by function ${\sf H}(U_0,\phi_0)$ is linear in terms of the size of the input because it contains a single pair $(U',\phi')$ for each subformula $\phi'$ of the original formula~$\phi_0$. The execution time of function ${\sf H}(U_0,\phi_0)$ is also linear.

Once the list of all required pairs $(U_1,\psi_1),\dots,(U_n,\psi_n)$ is identified, we use dynamic programming to fill-in  {\em two-dimensional} array $sat[w,i]$ in such a way that $sat[w,i]$ stores the value of $sat(w,U_i,\phi_i)$. This part of the code of our algorithm is given in Figure~\ref{algorithm}. Note that to improve the readability of the code, we write $sat[w,(U_i,\phi_i)]$ instead of $sat[w,i]$.

\begin{figure}[ht]
\vspace{-2mm}
\hrule\vspace{1mm}
\begin{algorithmic}
\footnotesize\sf



\For{$(U, \phi) \in {\sf H}(U_0, \phi_0)$}
    \For{$w\in W$}
            \Switch{$\phi$}
            \Case{$\phi$ is an atomic proposition} 
                \If{$(w, U)\in \pi(\phi)$} 
                    \State $sat[w,(U, \phi)]\gets true$
                \Else
                    \State $sat[w,(U, \phi)]\gets false$
                \EndIf
            \EndCase
            \Case{$\phi$ has the form $\neg\psi$} 
                \State $sat[w, (U, \phi)]\gets \neg sat[w,(U, \psi)]$
            \EndCase
            \Case{$\phi$ has the form $\psi\to\chi$} 
                \State \hspace{-1mm}$sat[w,\!(U, \phi)]\!\gets\!\neg sat[w,\!(U, \psi)] \vee sat[w,(U, \chi)]$
            \EndCase
            \Case{$\phi$ has the form $\B^T_X\psi$} 
                \State $sat[w,(U, \phi)] \gets true$
                \For{$w'\in W$}
                   \If{$w\sim_{U\cup X} w'$ and $T\subseteq \mathcal{T}_{w'}$ 
                        \State and $\neg sat[w',(U, \psi)]$}
                        \State $sat[w,(U, \phi)]\gets false$
                     \State \textbf{break}
                   \EndIf
                \EndFor
            \EndCase
            \Case{$\phi$ has the form $[X]\psi$}
                \State $sat[w,(U,\phi)]\gets sat[w,(U\cup X,\psi)]$
            \EndCase
            \EndSwitch
    \EndFor
\EndFor
\end{algorithmic}
\vspace{1mm}\hrule
\caption{Model checking algorithm.}\label{algorithm}
\vspace{-10mm}
\end{figure}

\section{Conclusion}\label{Conclusion section}

In this paper, we study the interplay between the trust-based belief modality and the public announcement modality. Although neither of these modalities is new, there are non-trivial properties that connect these modalities, and these properties have not been studied before. Our main results are a sound and complete logical system that includes a newly proposed Commutativity axiom and a non-trivial model checking algorithm. To prove the completeness of the logical system combining the two modalities, we significantly modify the semantics and existing proofs of completeness. Model checking for either modality has not been studied before.


\bibliographystyle{splncs04}
\bibliography{naumov}

\begin{thebibliography}{10}
\providecommand{\url}[1]{\texttt{#1}}
\providecommand{\urlprefix}{URL }
\providecommand{\doi}[1]{https://doi.org/#1}

\bibitem{abvs10jal}
{\AA}gotnes, T., Balbiani, P., van Ditmarsch, H., Seban, P.: Group announcement
  logic. Journal of Applied Logic  \textbf{8}(1),  62 -- 81 (2010).
  \doi{10.1016/j.jal.2008.12.002}

\bibitem{bv21jpl}
Baltag, A., van Benthem, J.: A simple logic of functional dependence. Journal
  of Philosophical Logic  \textbf{50},  1--67 (2021)

\bibitem{dhk07}
van Ditmarsch, H., van~der Hoek, W., Kooi, B.: Dynamic Epistemic Logic.
  Springer, Berlin, Germany (2007). \doi{10.1007/978-1-4020-5839-4}

\bibitem{vgw17icla}
van Eijck, J., Gattinger, M., Wang, Y.: Knowing values and public inspection.
  In: Indian Conference on Logic and Its Applications. pp. 77--90. Springer
  (2017)

\bibitem{fhv92jacm}
Fagin, R., Halpern, J.Y., Vardi, M.Y.: What can machines know? on the
  properties of knowledge in distributed systems. Journal of the ACM (JACM)
  \textbf{39}(2),  328--376 (1992)

\bibitem{f13wp}
Fisher, M.: Syrian hackers claim {A}{P} hack that tipped stock market by \$136
  billion. is it terrorism?
  \url{https://www.washingtonpost.com/news/worldviews/wp/2013/04/23/syrian-hackers-claim-ap-hack-that-tipped-stock-market-by-136-billion-is-it-terrorism/}
  (2013), {Accessed: 2022-05-31}

\bibitem{ga17tark}
Galimullin, R., Alechina, N.: Coalition and group announcement logic. In:
  Proceedings Sixteenth Conference on Theoretical Aspects of Rationality and
  Knowledge (TARK) 2017, Liverpool, UK, 24-26 July 2017. pp. 207--220 (2017)

\bibitem{gls15jair}
Grossi, D., Lorini, E., Schwarzentruber, F.: The ceteris paribus structure of
  logics of game forms. Journal of Artificial Intelligence Research
  \textbf{53},  91--126 (2015)

\bibitem{jn22ai}
Jiang, J., Naumov, P.: Data-informed knowledge and strategies. Artificial
  Intelligence  \textbf{309},  103727 (2022).
  \doi{10.1016/j.artint.2022.103727}

\bibitem{jn22ijcai-trust}
Jiang, J., Naumov, P.: In data we trust: The logic of trust-based beliefs. In:
  the 31st International Joint Conference on Artificial Intelligence
  ({I}{J}{C}{A}{I}-22) (2022)

\bibitem{jn24synthese-trust}
Jiang, J., Naumov, P.: A logic of trust-based beliefs. Synthese
  \textbf{204}(46) (2024)

\bibitem{l18mw}
Langlois, S.: This day in history: Hacked {A}{P} tweet about white house
  explosions triggers panic.
  \url{https://www.marketwatch.com/story/this-day-in-history-hacked-ap-tweet-about-white-house-explosions-triggers-panic-2018-04-23}
  (2018), {Accessed: 2022-05-31}

\bibitem{pgca19lori}
Perrotin, E., Galimullin, R., Canu, Q., Alechina, N.: Public group
  announcements and trust in doxastic logic. In: International Workshop on
  Logic, Rationality and Interaction. pp. 199--213. Springer (2019)

\bibitem{wa13synthese}
W{\'a}ng, Y.N., {\AA}gotnes, T.: Public announcement logic with distributed
  knowledge: expressivity, completeness and complexity. Synthese
  \textbf{190}(1),  135--162 (2013)

\end{thebibliography}

\end{document}